\documentclass[conference]{IEEEtran}
\IEEEoverridecommandlockouts
\usepackage{amsmath,amssymb,amsfonts,amsthm}
\usepackage{algorithmic}
\usepackage{graphicx}
\usepackage{textcomp}
\usepackage{xcolor}
\usepackage{multicol}
\usepackage{subcaption}
\def\BibTeX{{\rm B\kern-.05em{\sc i\kern-.025em b}\kern-.08em
    T\kern-.1667em\lower.7ex\hbox{E}\kern-.125emX}}

\usepackage[style=numeric-comp, sorting=none, doi=false, url=false, isbn=false, giveninits=true, natbib=true]{biblatex}
\usepackage{siunitx}
\usepackage{dsfont}

\DeclareMathOperator*{\argmin}{arg\,min}

\newcommand{\trace}[1]{\textup{tr}\!\left(#1\right)}
\newcommand{\real}[1]{\textup{Re}\!\left(#1\right)}
\newcommand{\imag}[1]{\textup{Im}\!\left(#1\right)}

\newtheorem{assumption}{Assumption}
\newtheorem{lemma}{Lemma}
\newtheorem{corollary}{Corollary}
\newtheorem{theorem}{Theorem}
\newtheorem{remark}{Remark}
\newtheorem{definition}[]{Definition}
\renewbibmacro{in:}{}
\AtEveryBibitem{%
  \clearfield{note}%
}

\begin{document}

\title{Performance Bounds for Quantum Feedback Control\\
\thanks{This work
was supported in part by the National Science Foundation under Grant
OAC-1835443, Grant SII-2029670, Grant ECCS-2029670, Grant OAC-
2103804, and Grant PHY-2021825, in part by the U.S. Agency for International
Development through Penn State under Grant S002283-USAID,
in part by the the MIT International Science and Technology Initiatives
through the MIT-Switzerland Lockheed Martin Seed Fund, in part by the
Advanced Research Projects Agency-Energy (ARPA-E), U.S. Department
of Energy under Grant DE-AR0001211 and Grant DE-AR0001222,
in part by the Research Council of Norway and Equinor ASA through
Research Council project under Grant 308817, in part by the United
States Air Force Research Laboratory, and in part by the the United
States Air Force Artificial Intelligence Accelerator and was accomplished
under Cooperative Agreement Number under Grant FA8750-19-2-1000.
The work of Julian Arnold was supported by the Swiss National Science under individual grant 200020\_200481.}
}

\author{
\IEEEauthorblockN{Flemming Holtorf\textsuperscript{\IEEEauthorrefmark{1}\IEEEauthorrefmark{2}}, Frank Sch{\"a}fer\textsuperscript{\IEEEauthorrefmark{1}\IEEEauthorrefmark{3}}, Julian Arnold\textsuperscript{\IEEEauthorrefmark{4}}, Christopher Rackauckas\textsuperscript{\IEEEauthorrefmark{1}}, Alan Edelman\textsuperscript{\IEEEauthorrefmark{1}}}
\IEEEauthorblockA{\textsuperscript{\IEEEauthorrefmark{1}} Massachusetts Institute of Technology, Cambridge, MA, USA}
\IEEEauthorblockA{\textsuperscript{\IEEEauthorrefmark{4}} University of Basel, Basel, Switzerland}
\IEEEauthorblockA{\textsuperscript{\IEEEauthorrefmark{2}} holtorf@mit.edu}
\IEEEauthorblockA{\textsuperscript{\IEEEauthorrefmark{3}} franksch@mit.edu}
}

\maketitle

\begin{abstract}
  The limits of quantum feedback control have immediate consequences for quantum information science at large, yet remain largely unexplored. Here we combine quantum filtering theory and moment-sum-of-squares techniques to construct a hierarchy of convex optimization problems that furnish monotonically improving, computable bounds on the best attainable performance for a broad class of quantum feedback control problems. These bounds may serve as witnesses of fundamental limitations, optimality certificates, or performance targets. We prove convergence of the bounds to the optimal control performance under technical conditions and demonstrate the practical utility of our approach by designing certifiably near-optimal controllers for a qubit in a cavity subjected to photon counting and homodyne detection measurements. 
\end{abstract}

\begin{IEEEkeywords}
Quantum Information \& Control, Stochastic Optimal Control, Quantum Filtering, Convex Optimization
\end{IEEEkeywords}

\section{Introduction}
Feedback control of devices at the quantum scale holds significant potential for current and future applications in the field of quantum information science~\cite{glaser2015training, koch2022quantum}. The nonlinear and stochastic nature of quantum systems under continuous observation, however, complicates the quest for an effective deployment of feedback control in practice~\cite{wiseman2009quantum}. While the characterization of optimal quantum feedback controllers through the dynamic programming principle has a longstanding history, dating back to the work of \citet{belavkin1988nondemolition} in 1988, designing such controllers by solving the quantum analog of the nonlinear Hamilton-Jacobi-Bellman (HJB) equation remains, barring a few simplified situations \cite{bouten2005bellman,belavkin2009dynamical,belavkin2006nondemolition}, an elusive challenge.
It is instead common practice to rely on heuristics, often rooted in reinforcement learning, gradient-based optimization, or expert intuition, to design quantum feedback controllers in applications~\cite{abdelhafez2019gradient, schaefer2021control, porotti2022deep}. And although such heuristically derived control policies are frequently found to perform remarkably well, their degree of suboptimality essentially always remains unquantified, leaving uncertainty about whether any observed performance limitations are fundamental or simply due to a suboptimal controller design. In this paper, we construct a method to bring clarity to such situations by computing informative bounds on the best attainable control performance for a wide range of quantum feedback control problems. These bounds may serve as certificates of optimality, witnesses of fundamental limitations, or performance targets and, as such, complement controller design heuristics.

The key insight to enable our construction is that viewing (stochastic) optimal control through the lens of linear programming over HJB subsolutions endows the problem with a convex, albeit infinite-dimensional, geometry~\cite{lions1982optimal,fleming1989convex,bhatt1996occupation}. From there, the moment-sum-of-squares hierarchy~\cite{parrilo2000structured,lasserre2001global} gives rise to a practical computational scheme; under mild assumptions, it generates a sequence of increasingly tight, tractable convex approximations to the original infinite-dimensional linear program which in turn furnishes a sequence of monotonically improving, practically computable bounds for the best attainable control performance. To apply this approach, which has already been considered for a range of classical control problems~\cite{lasserre2008nonlinear,savorgnan2009discrete,holtorf2024stochastic}, to the quantum realm, we leverage quantum filtering theory~\cite{belavkin1988nondemolition,belavkin1992quantum} and cast quantum feedback control problems as stochastic optimal control problems. The result is a method that enables the computation of informative bounds on the best attainable feedback control performance for a rich class of quantum systems under continuous observation via conic optimization.

The remainder of this article is organized as follows. First, Section~\ref{sec:notation} introduces our notation. Next, in Section~\ref{sec:control_problem}, we define the class of quantum feedback control problems under consideration and discuss key assumptions. Our main contribution, the construction and convergence analysis of a hierarchy of increasingly tight convex bounding problems for the best attainable quantum feedback control performance, is presented in Section~\ref{sec:bounding_problem}. Section~\ref{sec:extensions} briefly discusses extensions to this hierarchy. Finally, we demonstrate the practical utility of our bounding method with a qubit control example in Section~\ref{sec:example}, before we conclude in Section~\ref{sec:conclusion}.

\section{Notation}\label{sec:notation}
  Throughout this article, we rely on the following notational conventions.

  \textbf{Linear algebra \& analysis} -- The adjoint of a matrix (or vector) $A$ will be denoted by $A^*$. The commutator and anticommutator of two square matrices $A$ and $B$ will be denoted by $[A,B] = AB - BA$ and $\lbrace A, B\rbrace = AB + BA$, respectively. The notation $\langle \cdot, \cdot \rangle$ should not be confused with the Dirac notation commonly used in quantum physics but instead should be understood more broadly as the inner product between two (dual) vector spaces (most frequently spaces of Hermitian matrices). We use $\mathcal{C}^{1,2}(X \times Y)$ to denote functions with two arguments that are once, respectively twice, continuously differentiable with respect to their first and second argument on the domain $X \times Y$; when $X\times Y$ is closed, differentiability shall be understood in the sense of~\citet{whitney1992analytic}.
  
  \textbf{Probability} -- 
  For the sake of a light notation, we denote the classical expectation of a random variable $x$ by $\mathbb{E}[x]$ and omit explicit reflection of the underlying probability measure as that will be clear from context throughout. We use $\delta_x$ to refer to the Dirac measure at the singleton $\lbrace x\rbrace$. 

  \textbf{Algebraic geometry} -- The set of polynomials with real coefficients in the variables $x$ will be denoted by $\mathbb{R}[x]$; similarly, we refer to the restriction of $\mathbb{R}[x]$ to polynomials with degree at most $d$ with $\mathbb{R}_d[x]$. The set of sum-of-squares polynomials will be denoted by $\Sigma^2[x]$. Whenever we refer to polynomials in $\mathbb{R}[\rho]$ where $\rho \in \mathbb{C}^{n\times n}$, we mean a polynomial with real coefficients jointly in the elements of $\real{\rho}$ and $\imag{\rho}$. Lastly, we refer to vector- and matrix-valued functions as polynomials when all of their components are polynomials. 
  
\section{Quantum stochastic optimal control}\label{sec:control_problem}
We consider quantum systems with a Hermitian Hamiltonian of the form
\begin{align*}
  H(u) = H_0 + \sum_{k=1}^{K} u_k H_k,
\end{align*}
where $H_0$ denotes the drift Hamiltonian of the system, and $H_1, \dots, H_K$ are control fields with tunable drives $u = \begin{bmatrix} u_1 & \cdots &u_K \end{bmatrix}$. The control drives are further assumed to be confined to an admissible set $U \subset \mathbb{R}^K$. Finally, we assume that the systems under consideration are subjected to a continuous measurement process to enable feedback control. The dynamics of such systems, which we state below without further derivation, are established by quantum filtering theory. For an introductory treatment of the subject, we refer the reader to \cite{wiseman2009quantum,jacobs2006straightforward}.

The density matrix $\rho_t$ encoding the state of a quantum system conditioned on a continuous measurement process $\xi_t$ follows stochastic dynamics described by the Quantum Filtering Equation~\cite{belavkin1992quantum}
\begin{align}
   \mathrm{d} \rho_t = \mathcal{L}(u_t) \rho_t \, \mathrm{d}t + \mathcal{G} \rho_t \, \mathrm{d} \xi_t. \tag{QFE} \label{eq:QFE} 
\end{align} 
For systems of the described structure, the action of the Lindbladian $\mathcal{L}(u)$ is given by
\begin{align*}
  \mathcal{L}(u) \rho = -i [H(u), \rho] + \sum_{l=1}^L \left( \sigma_l \rho \sigma_l^* - \frac{1}{2}\lbrace \sigma_l^* \sigma_l, \rho \rbrace \right),
\end{align*}
where the jump operators $\sigma_l$ characterize the interaction between the quantum system and its environment due to observation. We focus on systems that are subjected to a combination of homodyne detection and photon counting measurements. For notational convenience, we partition the index set of measurements $\lbrace 1, \dots, L \rbrace$ into sets HD and PC, covering the homodyne detection and photon counting measurements, respectively. The innovation operator $\mathcal{G}$ then decomposes into two separate contributions according to the type of measurement:
\begin{align*}
  \mathcal{G} \rho_t  \, \mathrm{d} \xi_t = \sum_{l \in {\rm HD}} \mathcal{G}_l \rho_t \, \mathrm{d} w_t^l + \sum_{l \in {\rm PC}} \mathcal{G}_l \rho_t \, \mathrm{d} n_t^l -  \mathcal{L}_l\rho_t \, \mathrm{d}t.
\end{align*}
Homodyne detection causes diffusive innovations described by standard Gaussian increments $\mathrm{d} w_t^l$ which are in one-to-one correspondence with a measured homodyne current~\cite{wiseman2009quantum}. The associated innovation operator acts according to
\begin{align*}
  \mathcal{G}_l \rho_t \, \mathrm{d} w_t^l =  (\sigma_l \rho_t + \rho_t \sigma_l^* - \trace{\sigma_l \rho_t + \rho_t \sigma_l^*} \! \rho_t) \, \mathrm{d} w_t^l.
\end{align*}
Photon counting, in contrast, causes a deterministic drift
\begin{align*}
  \mathcal{L}_l \rho_t \, \mathrm{d}t = (\sigma_l \rho_t \sigma_l^* - \trace{\sigma_l \rho_t \sigma_l^*}\rho_t) \, \mathrm{d}t.
\end{align*} 
and leads to discrete innovations upon detection of photon emissions~\cite{wiseman2009quantum}. The emission of photons is described by Poisson counters $n_t^l$, which fire at rate $\lambda_l(\rho) = \trace{\sigma_l \rho \sigma_l^*}$, and conditioned on the measurement of a photon emission, the state of the system jumps to $h_l(\rho) = \sigma_l \rho \sigma_l^*/\text{tr}\left(\sigma_l \rho \sigma^*_l\right)$. The associated innovation operator accordingly acts as
\begin{align*}
  \mathcal{G}_l \rho_t \, \mathrm{d} n_t^l = \left( h_l(\rho_t) -\rho_t \right) \mathrm{d} n_t^l.
\end{align*}


As it will be relevant throughout, it is worth noting here that the dynamics described by \eqref{eq:QFE} inherently preserve purity of the (conditioned) quantum state.
\begin{lemma}\label{lem:purity_invariance}
  The set of pure quantum states
  \begin{align*}
      B = \left\{ \rho \in \mathbb{C}^{n\times n}: \rho^* = \rho, \trace{\rho} = \trace{\rho^2} =1 \right\}
  \end{align*} is invariant under the dynamics \eqref{eq:QFE}.
\end{lemma}
\begin{proof}
  Applying It{\^o}'s lemma to \eqref{eq:QFE} shows that \begin{align*}
  \mathrm{d} \! \left[\trace{\rho_t}\right] = \mathrm{d} \! \left[\trace{\rho_t^2}\right] = 0
  \end{align*}
  if $\rho_0 \in B$. Moreover, the right-hand side of \eqref{eq:QFE} maps Hermitian matrices into Hermitian matrices. 
\end{proof}

Given the described abstraction of feedback-controlled quantum systems under continuous observation, the goal of this paper is to bound the best attainable feedback control performance as defined by the Quantum Stochastic Optimal Control Problem 
\begin{align*}
  J^* = \inf_{u_t} \quad & \mathbb{E}\left[ \int_0^T \ell(\rho_t, u_t) \, \mathrm{d}t + m(\rho_T) \right] \tag{QSOCP} \label{QSOCP}\\
  \text{s.t.} \quad & \rho_t \text{ satisfies \eqref{eq:QFE} on } [0,T] \text{ with } \rho_0 \sim \nu_0, \\
              \quad & u_t \in U \text{ is non-anticipative on } [0,T],
\end{align*} 
where $\ell$ and $m$ encode the control performance in terms of accumulating stage cost and terminal cost, respectively. We wish to highlight that \eqref{QSOCP} includes many common quantum feedback control tasks as special cases. Several examples are listed in Table \ref{tab:examples}.

\begin{table}[]
  \begin{center}
  \renewcommand{\arraystretch}{1.3}
  \caption{Common quantum control tasks as special cases of \eqref{QSOCP}.}
  \label{tab:examples}
  \begin{tabular}{|l|c|c|}
       \hline
       Control task         & Stage cost $\ell(\rho_t, u_t)$ & Terminal cost $m(\rho_T)$                                               \\\hline
       State preparation   & $0$                 & $1-\psi_{\rm ref}^* \rho_T \psi^{\phantom{*}}_{\rm ref}$ \\
       State stabilization & $1-\psi_{\rm ref}^* \rho_t \psi^{\phantom{*}}_{\rm ref}$ & $0$ \\
       State purification  & $\trace{\rho_t^2}$  & $0$  \\
       Entanglement creation\textsuperscript{1} & $0$  & $1-\trace{\tilde{\rho}_T^2}$ \\
       \hline
       \hline
  \end{tabular}
  \end{center}
  \quad \footnotesize{\textsuperscript{1} $\tilde{\rho}_T$ denotes the reduced density matrix for a given subsystem~\cite{bhaskara2017generalized}}
\end{table}






In order to apply moment-sum-of-squares techniques to construct tractable lower bounding problems for \eqref{QSOCP}, we finally make the following assumptions.
\begin{assumption}\label{asmpt:initial_purity}
  The initial distribution $\nu_0$ of the quantum state satisfies $\text{supp} \, \nu_0 \subset B$, i.e., the initial state is guaranteed to be pure albeit potentially uncertain.
\end{assumption}

\begin{assumption}\label{asmpt:admissible_controls}
  The set of admissible control actions $U$ is a compact and basic closed semialgebraic set, i.e., there exist polynomials $\mathcal{U} = \lbrace q_1, \dots, q_r \rbrace$ such that $U = \{u \in \mathbb{R}^K : q(u) \geq 0, \forall q \in \mathcal{U}\}$ is compact. We refer to $\mathcal{U}$ as the control constraints. 
\end{assumption}

\begin{assumption}\label{asmpt:cost}
  The cost functions $\ell$ and $m$ are polynomials.
\end{assumption}

\begin{assumption}\label{asmpt:jump_operators}
  The jump operators $\sigma_l$ for $l \in {\rm PC}$ are such that $h_l(\rho)$ is a polynomial of degree at most one.    
\end{assumption}

Assumption \ref{asmpt:initial_purity} ensures that the quantum state, conditioned on the measurements, remains pure as per Lemma \ref{lem:purity_invariance} and thus confined to a basic closed semialgebraic set. Assumptions \ref{asmpt:admissible_controls} -- \ref{asmpt:jump_operators} guarantee further that the dynamics \eqref{eq:QFE} and control problem \eqref{QSOCP} are described entirely by polynomials. These properties will be essential for our construction of tractable bounding problems in Section \ref{sec:bounding_problem}.

We wish to emphasize that, while Assumptions \ref{asmpt:initial_purity} -- \ref{asmpt:cost} are extremely mild and may even be relaxed (see Section \ref{sec:extensions}), Assumption \ref{asmpt:jump_operators} is more limiting but still holds for many practically relevant photon counting measurement setups. Examples of photon counting measurements that satisfy Assumption \ref{asmpt:jump_operators} are measurements associated with unitary jump operators or measurements that cause a jump to the same quantum state independent of the state the photon emission occurred in.

\section{A convex bounding approach}\label{sec:bounding_problem}
To derive computable lower bounds on the optimal value of \eqref{QSOCP}, we draw inspiration from the dynamic programming heuristic. The dynamic programming heuristic asserts that the value function associated with \eqref{QSOCP}, i.e., the minimal cost-to-go
\begin{align}
  V(t, \rho) = \inf_{u_s} \quad & \mathbb{E}\left[ \int_t^T \ell(\rho_s, u_s) \, \mathrm{d}s + m(\rho_T) \right] \label{eq:val_fxn}\\
  \text{s.t.} \quad & \rho_s \text{ satisfies \eqref{eq:QFE} on } [t,T] \text{ with } \rho_t \sim \delta_{\rho}, \nonumber \\
              \quad & u_s \in U \text{ is non-anticipative on } [t,T], \nonumber
\end{align}
satisfies the Hamilton-Jacobi-Bellman (HJB) equation~\cite{oksendal2007applied}:
\begin{align*}
  \inf_{u \in U} \, &\mathcal{A}V(\cdot,\cdot,u) + \ell(\cdot, u) = 0 \text{ on } [0,T) \times B \\
  \text{s.t. } &V(T,\cdot) = m \text{ on } B. 
\end{align*}
Here, $\mathcal{A}$ refers to the infinitesimal generator~\cite{oksendal2007applied} associated with \eqref{eq:QFE}; the action of $\mathcal{A}$ on a smooth function $w\in\mathcal{C}^{1,2}([0,T]\times B)$ is given by 
\begin{align}
  \mathcal{A}w(t,\rho,u) = &\frac{\partial w}{\partial t}(t,\rho) + \langle \tilde{\mathcal{L}}(u)\rho  , \nabla_{\rho} w(t,\rho) \rangle \label{eq:inf_gen} \\
  &+ \frac{1}{2} \sum_{l\in {\rm HD}} \langle \mathcal{G}_l \rho, \nabla^2_{\rho} w(t,\rho) \, \mathcal{G}_l  \rho \rangle \nonumber\\
  &+ \sum_{l \in {\rm PC}} \lambda_l(\rho) \left(w(t, h_l(\rho)) - w(t,\rho) \right), \nonumber
\end{align}
where $\tilde{\mathcal{L}}(u) = \mathcal{L}(u) - \sum_{l\in PC} \mathcal{L}_l$ is the effective drift operator associated with a control action $u$. Note that, due to Assumption \ref{asmpt:initial_purity}, it suffices to solve the HJB equation on $[0,T] \times B$ as $B$ is invariant under the dynamics \eqref{eq:QFE} as per Lemma \ref{lem:purity_invariance}.

While the HJB equation is a nonlinear partial differential equation which is extremely difficult to solve even for low-dimensional systems, we may cast the search for a smooth HJB subsolution as a convex, albeit infinite-dimensional, optimization problem:
\begin{align*}
  \sup_{w \in \mathcal{C}^{1,2}([0,T]\times B)} \quad & \int_B w(0,\cdot) \, \mathrm{d} \nu_0  \tag{subHJB} \label{subHJB}\\
  \text{s.t.} \qquad \ \ \, \quad & \mathcal{A}w + \ell \geq 0 \text{ on }  [0,T] \times B \times U, \\
  & m - w(T,\cdot) \geq 0 \text{ on } B.
\end{align*}

\begin{lemma}\label{lem:underestimator}
  Any feasible point $w$ of \eqref{subHJB} underestimates the value function \eqref{eq:val_fxn} on $[0,T] \times B$ and so $\int_B w (0, \cdot) \, \mathrm{d} \nu_0$ underestimates the best attainable control performance $J^*$. 
\end{lemma}
\begin{proof}
 Feasibility of $w$ implies that $w(T,\cdot) \leq V(T,\cdot) = m$ on $B$. Now consider any time $0 \leq t < T$, any state $\rho \in B$, and any feedback controller $\{u_s\}_{s\in[t,T]}$ admissible on $[t,T]$. By feasibility of $w$, it follows that for $\rho_t \sim \delta_{\rho}$,
  \begin{align*}
      &\mathbb{E}\left[\int_{t}^T \ell(\rho_s, u_s) \, \mathrm{d}s + m(\rho_T) \right] \\
      & \geq \mathbb{E}\left[\int_{t}^T - \mathcal{A}w(s, \rho_s, u_s) \, \mathrm{d}s + w(T,\rho_T) \right]= w(t, \rho), 
  \end{align*}
  where we used Dynkin's formula~\cite{oksendal2003stochastic} in the last step. Finally, taking the infimum of the left-hand side over all admissible controllers establishes that $V(t, \rho) \geq w(t,\rho)$. 
\end{proof}
The infinite-dimensional nature of \eqref{subHJB} renders its immediate practical value rather limited. We therefore proceed by constructing tractable finite-dimensional restrictions of \eqref{subHJB} using the moment-sum-of-squares hierarchy. To that end, we restrict \eqref{subHJB} to optimization over polynomials of fixed maximum degree $d$ instead of arbitrary smooth functions and further strengthen the non-negativity constraints to sufficient sum-of-squares constraints. For this construction to be well-posed, we require that the left-hand side of the non-negativity constraints are polynomials and that non-negativity is imposed on closed basic semialgebraic sets. The latter is guaranteed by Assumption \ref{asmpt:admissible_controls} and the fact that $B$ is basic closed semialgebraic. The former follows from Assumptions \ref{asmpt:cost}, \ref{asmpt:jump_operators}, and the following result.
\begin{lemma}\label{lem:poly2poly}
  Under Assumption \ref{asmpt:jump_operators}, the infinitesimal generator $\mathcal{A}$ [cf. Eq.~\eqref{eq:inf_gen}] maps polynomials to polynomials.
\end{lemma}
\begin{proof}
  Let $w$ be a polynomial. Then, $\frac{\partial w}{\partial t}, \nabla_{\rho} w$, and $\nabla^2_{\rho} w$ are componentwise polynomials as polynomials are closed under differentiation. Further note that $\tilde{\mathcal{L}}(u)\rho$, $\mathcal{G}_l\rho$, $\lambda_l(\rho)$, and by Assumption \ref{asmpt:jump_operators} also $h_l(\rho)$, are componentwise polynomials. Since polynomials are also closed under addition, multiplication, and composition, it thus follows that $\langle \tilde{\mathcal{L}}(u) \rho , \nabla_{\rho} w(t,\rho) \rangle$, $\langle \mathcal{G}_l \rho, \nabla^2_{\rho} w(t,\rho) \, \mathcal{G}_l  \rho \rangle$, and $\lambda_l(\rho) \left(w(t, h_l(\rho)) - w(t,\rho)\right)$ are polynomials and therefore so is $\mathcal{A}w$.
\end{proof}

The resultant sum-of-squares restriction of \eqref{subHJB} reads
\begin{align*}
  J^*_d= \sup_{w_d \in \mathbb{R}_d[t,\rho]} \quad & \int_B w_d(0,\cdot) \, \mathrm{d} \nu_0  \tag{sosHJB$_d$} \label{sosHJB}\\
  \text{s.t.} \qquad \ \, & \mathcal{A}w_d + \ell \in Q_{d+2} \left[\mathcal{T}\cup \mathcal{B} \cup \mathcal{U}\right],\\
  & m - w_d(T,\cdot) \in Q_{d}\left[\mathcal{B}\right],
\end{align*}
where $Q_d[\mathcal{S}]$ denotes the bounded-degree quadratic modulus associated with a set of polynomials $\mathcal{S} = \{a_1,\dots, a_p\}$; formally, 
\begin{multline*}
  Q_d[\mathcal{S}] = \lbrace f \in \mathbb{R}[x] : f = s_0 + \sum_{i=1}^{p} s_i a_i, \\
  \text{ where } s_i \in \Sigma^2[x] \text{ with } \deg{s_i a_i} \leq d \rbrace.
\end{multline*} 
The set of control constraints $\mathcal{U}$ is defined as in Assumption \ref{asmpt:admissible_controls}. The sets $\mathcal{T}$ and $\mathcal{B}$ similarly denote collections of polynomial constraints that generate the sets $[0,T]$ and $B$, respectively. There is not a unique choice for these constraints.  We thus  make the following assumption, which, as will be shown, endows the corresponding sequence of performance bounds $J_d^*$ with favorable convergence guarantees. 
\begin{assumption}\label{asmpt:set_representation}
  For the construction of \eqref{sosHJB} we choose $\mathcal{T} = \lbrace t , T-t\rbrace $ so that $[0,T] = \lbrace t \in \mathbb{R} : p(t) \geq 0, \ \forall p \in \mathcal{T}\rbrace$. Moreover, to keep the computational burden associated with solving \eqref{sosHJB} at a minimum, we explicitly eliminate the symmetry and unit trace constraints in $B$ and represent density matrices only in terms of the remaining degrees of freedom, i.e., $\real{\rho_{ii}}$, for $1 \leq i < n$, and $\rho_{ij}$, for $1 \leq i < j \leq n$. The set of such reduced density matrix representations will be denoted by $D$. In the following, we abuse notation and refer to reduced density matrix representations simply by $\rho \in D$. In these reduced coordinates, the set of pure density matrices $B$ is given by a single polynomial equality constraint 
  \begin{multline}
       \trace{\rho^2} = \left(1- \sum_{i=1}^{n-1} \real{\rho_{ii}}\right)^2 + \sum_{i=1}^{n-1} \real{\rho_{ii}}^2 \\ + 2\sum_{1 \leq i < j \leq n} |\rho_{ij}|^2 = 1. \label{eq:purity}
  \end{multline} 
  Accordingly, we let $\mathcal{B} = \left\lbrace 1-\trace{\rho^2} ,\trace{\rho^2}-1 \right \rbrace$ so that $B = \lbrace \rho \in D: p(\rho) \geq 0, \ \forall p \in \mathcal{B} \rbrace$. 
\end{assumption}
The hierarchical structure of the bounding problems with respect to the degree parameter $d$ as described in the following corollary is desirable from a practical point of view as it allows us to trade off more computation for tighter bounds. 
\begin{corollary}\label{cor:sos_hierarchy}
  Any feasible point $w_d$ of Problem \eqref{sosHJB} underestimates the value function \eqref{eq:val_fxn} on $[0,T]\times B$ and $\int_B w_d(0,\cdot)\,\mathrm{d} \nu_0$  consequently underestimates $J^*$. Moreover, the optimal values $J_d^*$ form a monotonically increasing sequence.
\end{corollary}
\begin{proof}
  Any feasible point of \eqref{sosHJB} is also feasible for \eqref{subHJB} and thus underestimates the value function by Lemma \ref{lem:underestimator}. Since $Q_{d}[\mathcal{S}] \subset Q_{d+1}[\mathcal{S}]$ for any set of polynomials $\mathcal{S}$, it follows further that (sosHJB$_{d+1}$) is a relaxation of \eqref{sosHJB} and hence $J^*_{d+1} \geq J^*_{d}$.
\end{proof}

A natural question that arises from Corollary \ref{cor:sos_hierarchy} is if the bound $J_d^*$ approaches the true optimal value $J^*$ as the degree $d$ is increased. In the following, we make a first step toward analyzing this convergence question. Specifically, we prove convergence whenever \eqref{QSOCP} admits a smooth value function and the control constraints satisfy the following mild regularity condition.
\begin{definition}[Putinar's Condition~\cite{putinar1993positive}]
  We say a set of polynomials $\mathcal{S} \subset \mathbb{R}[x]$ satisfies Putinar's condition if $\exists N > 0$ such that $N - \sum_{i=1}^n x_i^2 \in Q_{d}[\mathcal{S}]$ for some positive integer $d$.
\end{definition}
To that end, we first observe that the polynomials that frame Problem \eqref{sosHJB} naturally satisfy Putinar's condition as long as the control constraints do.
\begin{lemma}\label{lem:putinar_condition}
  The set $\mathcal{B}$ as defined in Assumption \ref{asmpt:set_representation} satisfies Putinar's condition. If further the set of control constraints $\mathcal{U}$ satisfies Putinar's condition, then so does the set $\mathcal{T}\cup \mathcal{B} \cup \mathcal{U}$.
\end{lemma}
\begin{proof}
   From the description in Assumption \ref{asmpt:set_representation}, it is easily verified that $\mathcal{B}$ satisfies Putinar's condition since Eq. \eqref{eq:purity} yields for $\rho \in D$ that
   \begin{multline*}
       1 - \sum_{i=1}^{n-1} \real{\rho_{ii}}^2 - \sum_{1\leq i < j \leq n} |\rho_{ij}|^2 =\\
       \left(1- \sum_{i=1}^{n-1} \real{\rho_{ii}}\right)^2 + \sum_{1\leq j < j \leq n} |\rho_{ij}|^2  + 1-\trace{\rho^2}.
   \end{multline*}
   The right-hand side of the relation above is clearly an element of $Q_{2}[\mathcal{B}]$ as the first two terms are sums of squares. Further, any set of degree-one polynomials defining a bounded polyhedron satisfies Putinar's condition~\cite{lasserre2010moments}, thus $\mathcal{T}$ does as well. Finally note that 
   $a \in Q_d[\mathcal{T}]$, $b \in Q_d[\mathcal{B}]$, $c \in Q_d[\mathcal{U}]$ implies that $a + b + c \in Q_d[\mathcal{T}\cup\mathcal{B}\cup \mathcal{U}]$ as $\mathcal{T}$, $\mathcal{B}$, and $\mathcal{U}$ are comprised of polynomials in distinct variables. The conclusion follows. 
\end{proof}
With this in hand, the convergence of the bounds furnished by \eqref{sosHJB} can be established by application of Putinar's Positivstellensatz \cite{putinar1993positive} according to the following theorem.
\begin{theorem}\label{thm:convergence}
  If the value function \eqref{eq:val_fxn} is $\mathcal{C}^{1,2}([0,T]\times B)$ and the set of control constraints $\mathcal{U}$ satisfies Putinar's condition, then $J^*_d \uparrow J^*$.
\end{theorem}
\begin{proof}
  Let $\epsilon > 0$ and recall that on a compact set any continuously differentiable function and its (partial) derivatives can be approximated uniformly by a polynomial and its derivatives~\cite{whitney1992analytic}. Therefore, there exists a polynomial $w$ such that 
  \begin{align*}
      \|V - w \|_{\infty}, \| \mathcal{A}V - \mathcal{A}w\|_{\infty} < \epsilon,
  \end{align*} 
  where $\|\cdot \|_{\infty}$ refers to the sup norm on the associated domains $[0,T]\times B$ and $[0,T]\times B \times U$, respectively. 
  Under the assumed smoothness of the value function $V$, it is well-known that $V$ satisfies the HJB equation (see e.g. \cite[Thm. 3.1]{oksendal2007applied}) and thus it holds in particular that
  \begin{align*}
      &\mathcal{A}V + \ell \geq 0 \text{ on } [0,T] \times B\times U,\\
      &m - V(T,\cdot) \geq 0 \text{ on } B.
  \end{align*}
  Now consider $\hat{w} = w + 2\epsilon (t-T-1)$ and note that, by construction, $\mathcal{A}\hat{w} = \mathcal{A}w+2\epsilon$ and $\hat{w}(T,\cdot) = w(T, \cdot)-2\epsilon$. It follows that
  \begin{align*}
      &\mathcal{A}\hat{w} + \ell \geq \mathcal{A}V + \ell + \epsilon > 0 \text{ on } [0,T] \times B \times U, \\
      &m - \hat{w}(T, \cdot) \geq m - V(T,\cdot) + \epsilon > 0 \text{ on } B.
  \end{align*}
  By Lemma \ref{lem:putinar_condition}, Putinar's Positivstellensatz~\cite[Lemma 4.1]{putinar1993positive} therefore guarantees for sufficiently large $d$ that $\mathcal{A}\hat{w} + \ell \in Q_{d+2}[\mathcal{T}\cup\mathcal{B}\cup\mathcal{U}]$ and likewise $m - \hat{w}(T,\cdot) \in Q_{d}[\mathcal{B}]$ such that $\hat{w}$ is feasible for \eqref{sosHJB}. 
  The result follows by noting that
  \begin{align*}
      J^* - J_d^* &\leq \int_B \left\lvert V(0,\cdot)-\hat{w}(0,\cdot) \right\rvert \, \mathrm{d} \nu_0\\
      &\leq \max_{\rho \in B} |V(0,\rho) - w(0,\rho)| + | 2\epsilon(T + 1) |\\
      &<  (2T+3)\epsilon.
  \end{align*}
\end{proof}
\begin{remark}
  It should be emphasized that the assumption that \eqref{QSOCP} admits a smooth value function is by no means weak and, even if satisfied, generally not easily verified. Theorem \ref{thm:convergence} is only a first step toward establishing a formal basis for our empirical observation that the bounds furnished by \eqref{sosHJB} often appear to converge to $J^*$ in practice. Related work~\cite{fleming1989convex,bhatt1996occupation,lasserre2008nonlinear} suggests that the conditions under which convergence can be guaranteed may be substantially relaxed. 
\end{remark}

Finally we wish to emphasize that \eqref{sosHJB} is equivalent to a semidefinite program (SDP)~\cite{parrilo2000structured,lasserre2001global} which is readily and automatically constructed in optimization modeling tools~\cite{dunning2017jump, weisser2019polynomial, legat2022mathoptinterface} available in modern programming languages like Julia~\cite{bezanson2017julia}. The resultant SDPs may then be solved with a range of powerful off-the-shelf available solvers~\cite{andersen2000mosek, toh1999sdpt3, sturm1999using, odonoghue2016conic, garstka2021cosmo, coey2022solving}.

\section{Extensions}\label{sec:extensions}
\subsection{Infinite horizon problems}
While we detailed our analysis for the finite horizon problem \eqref{QSOCP}, one can construct analogous bounding problems for (discounted) infinite horizon problems. To that end, suppose our control objective is of the form
\begin{align*}
  \mathbb{E} \left[ \int_0^\infty e^{-\gamma t} \ell(\rho_t, u_t) \, \mathrm{d}t \right]
\end{align*}
with discount rate $\gamma > 0$. Then, one may notice from Eq. \eqref{eq:inf_gen} that 
\begin{align*}
  \mathcal{A}(e^{-\gamma t} w) = e^{-\gamma t} \left(  \mathcal{A}w - \gamma w \right).
\end{align*}
It follows by analogous arguments as in the proof of Lemma \ref{lem:underestimator} that for any $w\in \mathcal{C}^{1,2}([0,\infty)\times B)$ that satisfies
\begin{align*}
  &\mathcal{A}w - \gamma w + \ell \geq 0 \text{ on } [0,\infty) \times B \times U,
\end{align*}
$e^{-\gamma t}w$ is a global underestimator of the value function associated with the infinite horizon problem.
A simple modification of the bounding problems \eqref{sosHJB} thus enables the tractable computation of monotonically improving value function underestimators and hence performance bounds in the case of an infinite control horizon. 

\subsection{Mixed initial states}\label{sec:mixed_states}
The relaxation of Assumption \ref{asmpt:initial_purity} to mixed initial quantum states is possible at the expense of introducing additional conservatism. For initially mixed quantum states, \eqref{subHJB} characterizes valid bounds when the constraints are enforced on the set of all mixed quantum states
\begin{align*}
  \bar{B} = \left\lbrace \rho \in \mathbb{C}^{n\times n} : \rho = \rho^*, \trace{\rho} =1, \psi^* \rho \psi \geq 0 \ \forall \psi \in \mathbb{C}^n\right\rbrace.
\end{align*}
As $\bar{B}$ can be represented by finitely many polynomial inequality constraints (by Sylvester's criterion), the resultant problem in principle admits valid sum-of-squares restrictions analogous to \eqref{sosHJB}. However, these restrictions are impractical. They rely on the reformulation of the positivity requirement of a mixed quantum state in terms of the non-negativity of all its principal minors. The large number and high degree of these polynomial inequalities render the corresponding sum-of-squares restrictions expensive to solve. A more practical approach is to impose the constraints in \eqref{subHJB} instead on a simpler closed basic semialgebraic overapproximation of $\bar{B}$; for example,
\begin{align*}
  \tilde{B} =\left\lbrace \rho \in \mathbb{C}^{n\times n} : \rho^* = \rho, \trace{\rho} = 1, \trace{\rho^2} \leq 1\right\rbrace.
\end{align*}
This modification potentially introduces additional conservatism as $\tilde{B}$ is a strict superset of $\bar{B}$ but leads to more practical sum-of-squares restrictions. Note that the restriction $\tilde{B}$ can be refined flexibly by imposing additional constraints of the form $\psi^* \rho \psi \geq 0$ for any number of fixed state vectors $\psi \in \mathbb{C}^n$.



\subsection{Imperfect measurements}
So far we have implicitly assumed {\em lossless} or {\em perfect} homodyne and photon counting measurements. In practice, however, various factors can lead to {\em imperfect} detection \cite{wiseman2009quantum}. While in most such cases the bounds furnished by \eqref{sosHJB} will remain valid due to the simple fact that additional losses typically lead to more stringent performance limitations, it is often of interest to quantify explicitly the limitations induced by measurement imperfections. The presented bounding method extends naturally to this task. To that end, it is necessary to account for measurement inefficiencies in \eqref{eq:QFE}. For homodyne measurements with efficiency $\eta \in [0,1]$, the innovation operator in \eqref{eq:QFE} acts according to  
\begin{align*}
  \mathcal{G}_l \rho = \eta (\sigma_l \rho_t + \rho \sigma_l^* - \trace{\sigma_l \rho_t + \rho \sigma_l^*} \rho ), 
\end{align*}
and for inefficient photon detection, the drift operator $\mathcal{L}_l$ must be modified to
\begin{align*}
  \mathcal{L}_l \rho = \eta \left(\sigma_l \rho \sigma_l^* - \trace{\sigma_l \rho \sigma_l^*} \rho_t \right)
\end{align*}
and the rate of the driving Poisson counter decays to $\lambda_l(\rho) = \eta \trace{\sigma_l \rho \sigma_l^*}$ \cite[Section 4.8]{wiseman2009quantum}. It is easily observed that under these modifications the conclusion of Lemma \ref{lem:poly2poly} remains valid. Given imperfect measurements ($\eta < 1$), however, the conclusion of Lemma \ref{lem:purity_invariance} no longer holds and pure quantum states no longer remain pure under the resultant dynamics. As a consequence, the set of pure states $B$ in \eqref{subHJB} must be replaced by a basic semialgebraic overapproximation of the set of mixed states as discussed in Section \ref{sec:mixed_states}. 



\subsection{Extraction of heuristic controllers} \label{sec:controllers}
Bounds computed via \eqref{sosHJB} may be used to certify the near-optimality of any given control policy. As such, the proposed bounding method complements heuristic approaches for the design of control policies. The solution of \eqref{sosHJB}, however, can also be used to inform controller design directly. At the optimal point of \eqref{sosHJB}, the optimization variable $w_d$ approximates by construction the best possible polynomial underapproximator of the value function. Thus, it is reasonable to use $w_d$ as a proxy for the value function~\cite{henrion2008nonlinear,savorgnan2009discrete} and construct a heuristic controller by greedily descending on $w_d$, i.e., 
\begin{align}
  u^*_t(\rho) \in \argmin_{u\in U} \mathcal{A}w_d(t,\rho,u) + \ell(\rho,u). \label{eq:DPHeuristic}
\end{align}
The above requires minimization of a polynomial over $U$, which is only expected to be tractable in the case of one or few control inputs. Otherwise, we argue that the inherently heuristic nature of this construction may justify the use of fast heuristics to find local or approximate minimizers instead, for example by relying on recent advances in machine learning~\cite{deits2019lvis, bertsimas2021voice}. 

\section{Example}\label{sec:example}
We finally demonstrate the utility of the proposed bounding method for the problem of stabilizing the state of a qubit in a cavity~\cite{schaefer2021control}. Figure \ref{fig:sys} illustrates the system under consideration.
\begin{figure}[h!]
  \centering
  \includegraphics[scale = 0.5]{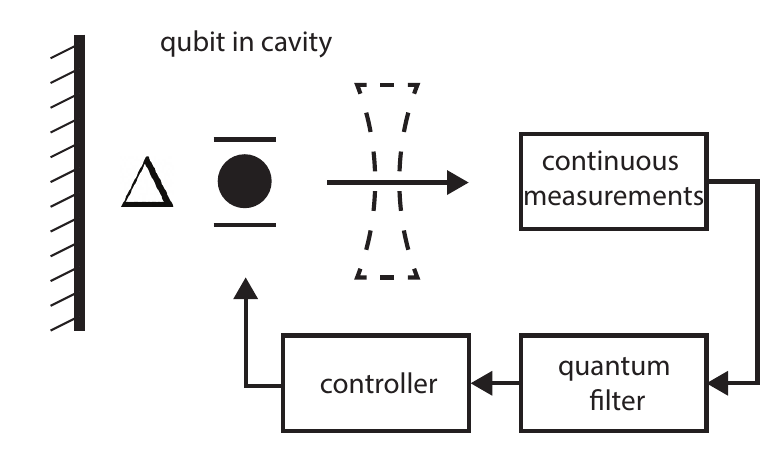}
  \caption{Control loop: qubit in a cavity subjected to continuous measurements.}
  \label{fig:sys}
\end{figure}
The Hamiltonian of the qubit is given by 
\begin{align*}
  H(u) = \frac{\Delta}{2} \sigma_z + \frac{\Omega}{2} u \sigma_x,
\end{align*} where $\sigma_x$ and $\sigma_z$ denote the Pauli matrices:
\begin{align*} 
  \sigma_x = \begin{bmatrix}
                  0 & 1 \\
                  1 & 0
              \end{bmatrix} \text{ and } \sigma_z = \begin{bmatrix} 
                  1 & 0 \\
                  0 & -1
              \end{bmatrix}.
\end{align*}    
To enable feedback, we assume that the qubit is subjected to continuous measurements associated with the jump operator
\begin{align*} 
  \sigma = \kappa \begin{bmatrix}
      0 & 0 \\
      1 & 0
  \end{bmatrix}.
\end{align*}
Note that such a measurement conforms with Assumption \ref{asmpt:jump_operators}. The parameters are chosen as $\Delta = \Omega = 5$ and $\kappa = 1$; the set of admissible control actions is $U=[-1,1]$. In the following, we consider a realization of the measurements through homodyne detection and photon counting setups and contrast the two. The objective of the control problem is to stabilize the excited state $\psi^{\phantom{*}}_{\rm ref} = \begin{bmatrix} 1 & 0 \end{bmatrix}^*$ with minimal expected infidelity
\begin{align*}
  \mathbb{E} \left[ \int_0^T 1 - \psi_{\rm ref}^* \rho_t \psi^{\phantom{*}}_{\rm ref} \, \mathrm{d}t \right]
\end{align*} (viz. maximum expected fidelity). The qubit is assumed to reside initially in its ground state $\psi_0 = \begin{bmatrix}
  0 & 1 
\end{bmatrix}^*$ and the distribution of the initial state is hence given by $\nu_0 = \delta_{\psi^{\phantom{*}}_0 \psi_0^*}$.

For the implementation of the presented bounding method, we relied on the optimization ecosystem in Julia. Specifically, we used \texttt{MarkovBounds.jl}~\cite{holtorf2024stochastic} and \texttt{SumOfSquares.jl}~\cite{weisser2019polynomial} to assemble the bounding problems and pass the resultant SDPs via the \texttt{MathOptInterface}~\cite{legat2022mathoptinterface} to Mosek v10~\cite{andersen2000mosek}. All computations were performed on a MacBook M1 Pro with 16GB unified memory. 

Table \ref{tab:performance_bounds} summarizes upper bounds for the maximal average fidelity attainable with both measurement setups as obtained by solving the bounding problems \eqref{sosHJB} for increasing degree $d$. %
\begin{table}[h]
  \centering
  \renewcommand{\arraystretch}{1.2}

  \caption{Performance bounds for feedback controlled qubit}
  \label{tab:performance_bounds}
  \textbf{Homodyne detection}
  \vspace{0.5em}
  
  \begin{tabular}{|c|c|c|}
      \hline
      Degree $d$ & Fidelity bound & Computational time [$\si{\second}$] \\ \hline 
      2 & 0.8502 & 0.008 \\ 
      4 & 0.8111 & 0.078 \\ 
      6 & 0.7973 & 0.64 \\ 
      8 & 0.7893 & 5.0 \\ 
      10 & 0.7856 & 27.9 \\
      \hline
      \hline
      \multicolumn{2}{|l}{\textbf{Best known fidelity:} 0.7750} & \\
      \hline
  \end{tabular}
  \vspace{1.5em}
  
  \textbf{Photon counting}
  \vspace{0.5em}

  \begin{tabular}{|c|c|c|}
      \hline
      Degree $d$ & Fidelity bound & Computational time [$\si{\second}$] \\ \hline 
      2 & 0.9602 & 0.0043 \\ 
      4 & 0.7497 & 0.031 \\ 
      6 & 0.7153 & 0.180 \\ 
      8 & 0.6902 & 1.67 \\ 
      10 & 0.6798 & 14.9 \\
      \hline
      \hline
      \multicolumn{2}{|l}{\textbf{Best known fidelity:} 0.6547} & \\
      \hline
  \end{tabular}
\end{table}
\begin{figure*}[h!]
  \centering
  \includegraphics[width=0.6\linewidth]{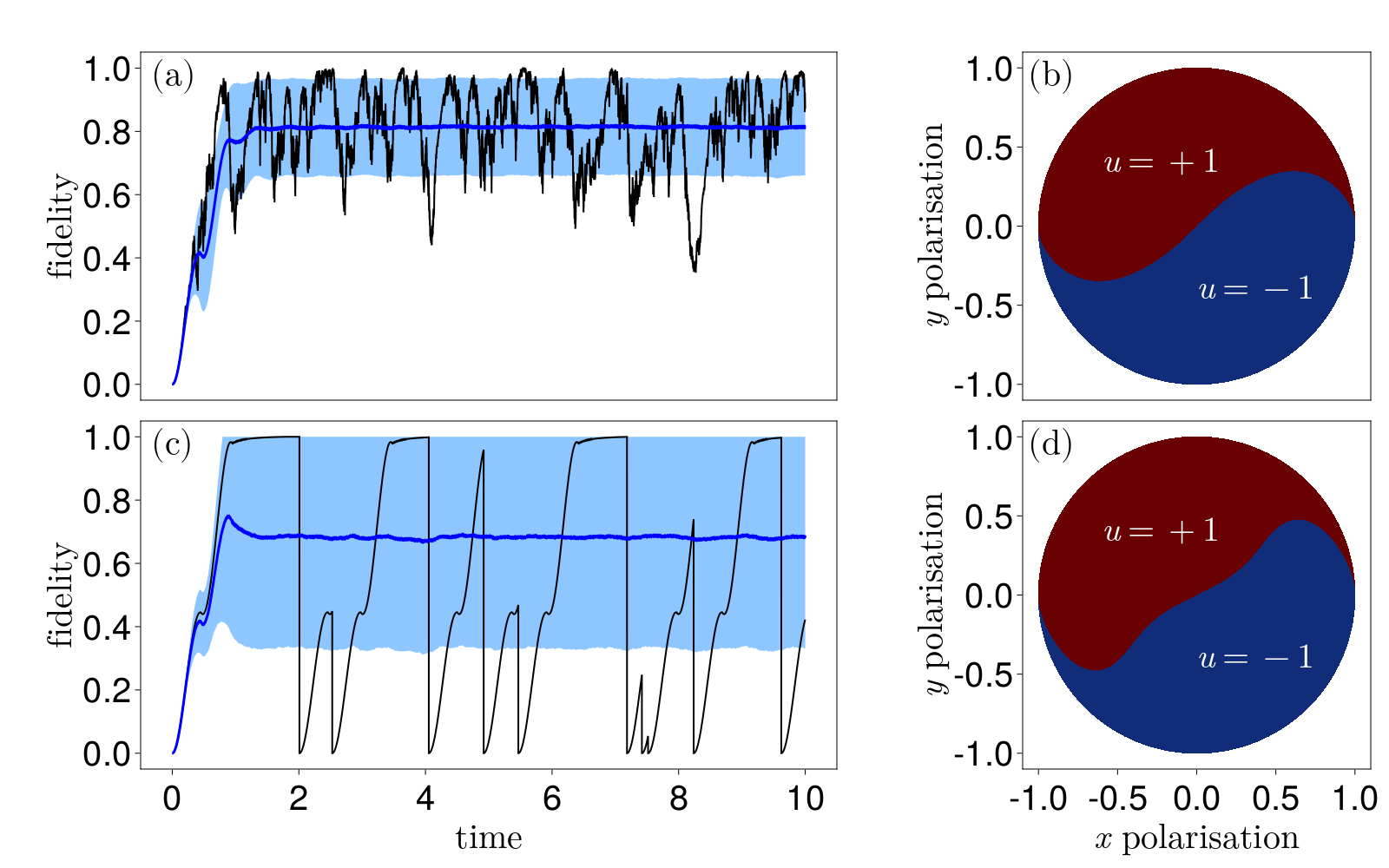}
\caption{Fidelity of a closed-loop controlled qubit alongside a visualization of the heuristic controller at $t=5$ in the x-y plane of the Bloch sphere for homodyne (a,b) and photon counting measurements (c,d). Mean trace and standard deviation band are shown in blue. A representative sample path is shown in black. \label{fig:res}}
\end{figure*}%
The bounds are clearly non-trivial and suggest to be informative even for moderate degrees. To emphasize this point, we further constructed heuristic controllers for both measurement setups from the optimal solution of $(\text{sosHJB}_4)$ as described in Section \ref{sec:controllers}. Their empirical performance serves as an achievable lower bound for the best attainable mean fidelity. Figure~\ref{fig:res} shows the mean fidelity and noise level attained by both controllers, alongside a visualization of the associated control policy as a function of the polarisations of the quantum state. The controllers achieve mean fidelities of $\SI{77.50}{\percent}$ and $\SI{65.47}{\percent}$ (ensemble averages over 10,000 sample trajectories) for the homodyne detection and photon counting setup, respectively. Against the backdrop of the computed bounds, the controllers are thus certifiably near-optimal, showcasing the practical utility of the proposed bounding method. An interesting spillover of this example is that, barring (highly unlikely) major statistical errors in the estimates of the fidelity attained by the heuristic controllers, this case study constitutes a computational proof that under the assumed circumstances a homodyne detection setup allows for strictly and significantly greater average mean fidelity than photon counting. This demonstrates that the proposed bounding method may provide relevant insights for the design of quantum devices at an early stage.

\section{Conclusion}\label{sec:conclusion}
Using quantum filtering theory and moment-sum-of-squares techniques, we have devised a hierarchy of convex optimization problems that furnishes a sequence of monotonically improving, practically computable bounds for the best attainable feedback control performance for a broad class of quantum systems subjected to continuous measurements. The bound sequence is proved to converge to the true optimal control performance under technical conditions. As demonstrated for a qubit in a cavity, we argue that the proposed bounding method can have relevant implications for the design of controlled quantum devices. On the one hand, it provides access to heuristic controllers alongside performance bounds which can guide controller design or certify the optimality of a given control policy. On the other hand, the bounds may serve as witnesses of fundamental limitations and thus inform the design of quantum systems at an early stage.

\section*{Acknowledgment}
We thank Gaurav Arya and Aparna Gupte for helpful discussions. The views and opinions of authors expressed herein 
do not necessarily state or reflect those of the United StatesGovernment
or any agency thereof. The views and conclusions contained in this 
document are those of the authors and should not be interpreted as 
representing the official policies, either expressed or implied, of the 
United States Air Force or the U.S. Government. The U.S. Government
is authorized to reproduce and distribute reprints for Government purposes notwithstanding any copyright notation herein.

\AtNextBibliography{\scriptsize}
\printbibliography

\end{document}